\documentclass[a4paper,11pt]{fullverllncs}
\usepackage[left=2.5cm,top=2.5cm,right=2.5cm,centering]{geometry}

\usepackage{graphicx}

\usepackage{amsmath,amssymb}
\usepackage{graphicx}
\usepackage{ascmac}
\usepackage{array}
\usepackage{algpseudocode}
\usepackage{amsfonts}
\usepackage{amssymb}
\usepackage{amsmath}
\usepackage{algorithm, algpseudocode}
\usepackage{multirow}
\usepackage{arydshln}
\usepackage{hhline} 
\usepackage[misc,geometry]{ifsym} 

\usepackage[dvipdfmx, colorlinks]{hyperref, xcolor}
\definecolor{winered}{rgb}{0.5,0,0}
\definecolor{darkblue}{rgb}{0,0,0.5}
\definecolor{darkgreen}{rgb}{0,0.3,0}
\hypersetup{
linkcolor=winered,
citecolor=darkblue,
urlcolor=darkgreen
}

\newcommand{\A}{\mathsf{A}}
\newcommand{\B}{\mathsf{B}}
\newcommand{\C}{\mathsf{C}}

\newcommand{\Adv}{\mathsf{Adv}}

\newcommand{\Z}{\mathbb{Z}}
\newcommand{\G}{\mathbb{G}}
\newcommand{\N}{\mathbb{N}}

\newcommand{\Inst}{\mathtt{inst}}
\newcommand{\instsign}{\mathtt{sign}}
\newcommand{\instskip}{\mathtt{skip}}
\newcommand{\accept}{\mathtt{accept}}
\newcommand{\reject}{\mathtt{reject}}

\newcommand{\DS}{\mathsf{DS}}
\newcommand{\vk}{\mathsf{vk}}
\newcommand{\sk}{\mathsf{sk}}

\newcommand{\Setup}{\mathsf{Setup}}
\newcommand{\KeyGen}{\mathsf{KeyGen}}
\newcommand{\Sign}{\mathsf{Sign}}
\newcommand{\Verify}{\mathsf{Verify}}
\newcommand{\EUFCMA}{\mathsf{EUF \mathchar`- CMA}}
\newcommand{\OTEUFCMA}{\mathsf{OT \mathchar`-  EUF \mathchar`- CMA}}
\newcommand{\rmEUFCMA}{\mathrm{EUF}\mathchar`-\mathrm{CMA}}
\newcommand{\rmOTEUFCMA}{\mathrm{OT}\mathchar`-\mathrm{EUF}\mathchar`-\mathrm{CMA}}
\newcommand{\Cert}{\mathsf{Cert}}

\newcommand{\rmoMSDHii}{1 \mathrm{ \mathchar`-MSDH\mathchar`- 2}}
\newcommand{\rmqMSDHii}{\mathit{q} \mathrm{ \mathchar`-MSDH\mathchar`- 2}}

\newcommand{\CL}{\mathsf{CL}}

\newcommand{\MCL}{\mathsf{MCL}}
\newcommand{\MCLSetup}{\mathsf{MCL.Setup}}
\newcommand{\MCLKeyGen}{\mathsf{MCL.KeyGen}}
\newcommand{\MCLSign}{\mathsf{MCL.Sign}}
\newcommand{\MCLVerify}{\mathsf{MCL.Verify}}

\newcommand{\SyncAS}{\mathsf{SAS}}
\newcommand{\SyncASSetup}{\mathsf{SAS.Setup}}
\newcommand{\SyncASKeyGen}{\mathsf{SAS.KeyGen}}
\newcommand{\SyncASSign}{\mathsf{SAS.Sign}}
\newcommand{\SyncASVerify}{\mathsf{SAS.Verify}}
\newcommand{\SyncASAggregate}{\mathsf{SAS.Aggregate}}
\newcommand{\SyncASAggVerify}{\mathsf{SAS.AggVerify}}

\newcommand{\MCLSyncAS}{\mathsf{SAS_{LLY}}}
\newcommand{\MCLSyncASSetup}{\mathsf{SAS_{LLY}.Setup}}
\newcommand{\MCLSyncASKeyGen}{\mathsf{SAS_{LLY}.KeyGen}}
\newcommand{\MCLSyncASSign}{\mathsf{SAS_{LLY}.Sign}}
\newcommand{\MCLSyncASVerify}{\mathsf{SAS_{LLY}.Verify}}
\newcommand{\MCLSyncASAggregate}{\mathsf{SAS_{LLY}.Aggregate}}
\newcommand{\MCLSyncASAggVerify}{\mathsf{SAS_{LLY}.AggVerify}}

\spnewtheorem{assumption}{Assumption}{\bfseries}{\itshape}

\begin{document}



\title{Improved Security Proof for the Camenisch-Lysyanskaya Signature-Based Synchronized Aggregate Signature Scheme\thanks{A preliminary version \cite{TT20} of this paper is appeared in Information Security and Privacy - 25th Australasian Conference (ACISP 2020).}}
\author{Masayuki Tezuka\textsuperscript{(\Letter)} \and Keisuke Tanaka}
\authorrunning{M.Tezuka et al.}
\institute{Tokyo Institute of Technology, Tokyo, Japan\\
\email{tezuka.m.ac@m.titech.ac.jp}}

\maketitle              
\pagestyle{plain}
\noindent
\makebox[\linewidth]{March 15, 2023}

\begin{abstract}

The Camenisch-Lysyanskaya signature scheme in CRYPTO 2004 is a useful building block to construct privacy-preserving schemes such as anonymous credentials, group signatures or ring signatures.
However, the security of this signature scheme relies on the interactive assumption called the LRSW assumption.
Even if the interactive assumptions are proven in the generic group model or bilinear group model, the concerns about these assumptions arise in a cryptographic community.
This fact caused a barrier to the use of cryptographic schemes whose security relies on these assumptions.

Recently, Pointcheval and Sanders proposed the modified Camenisch-Lysyanskaya signature scheme in CT-RSA 2018.
This scheme satisfies the EUF-CMA security under the new $q$-type assumption called the Modified-$q$-Strong Diffie-Hellman-2 ($\rmqMSDHii$) assumption.
However, the size of a $q$-type assumptions grows dynamically and this fact leads to inefficiency of schemes.

In this work, we revisit the Camenisch-Lysyanskaya signature-based synchronized aggregate signature scheme in FC 2013.
This scheme is one of the most efficient synchronized aggregate signature schemes with bilinear groups.
However, the security of this synchronized aggregate scheme was proven under the one-time LRSW assumption in the random oracle model.
We give the new security proof for this synchronized aggregate scheme under the $\rmoMSDHii$ (static) assumption in the random oracle model with little loss of efficiency.

\end{abstract}

\section{Introduction}
\subsection{Background}\label{IntroBack}
\paragraph{\bf{Aggregate Signatures.}}
Aggregate signature schemes originally introduced by Boneh, Gentry, Lynn, and Shacham \cite{BGLS03} allow anyone to convert $n$ individual signatures $(\sigma_1, \dots, \sigma_n)$ produced by different $n$ signers on different messages into the aggregate signature $\Sigma$ whose size is much smaller than a concatenation of the individual signatures.

This feature leads significant reductions of bandwidth and storage space in BGP (Border Gateway Protocol) routing \cite{BGOY07,BGLS03,LOSSW06}, bundling software updates \cite{AGH10}, sensor network data \cite{AGH10}, authentication \cite{OBY19}, and blockchain protocol \cite{HW18,SMD14,Zhao19}.

After the introduction of aggregate signatures, various aggregate signatures have been proposed: sequential aggregate signatures \cite{LMRS04}, identity-based aggregate signatures \cite{GR06}, synchronized aggregate signatures \cite{AGH10,GR06}, and fault-tolerant aggregate signatures \cite{HKKKR16}.

\paragraph{\bf{Synchronized Aggregate Signatures.}}
Synchronized aggregate signatures are a special type of aggregate signatures.
The concept of the synchronized setting aggregate signature scheme was introduced by Gentry and Ramzan \cite{GR06}.

Ahn, Green, and Hohenberger~\cite{AGH10} revisited the Gentry-Ramzan model and formalized the synchronized aggregate signature scheme.
In this scheme, all of the signers have a synchronized time period $t$ and each signer can sign a message at most once for each period~$t$.
A set of signatures that are all generated for the same period $t$ can be aggregated into a short signature.

It is useful to adopt synchronized aggregate signature schemes to systems which have a natural reporting period, such as log or sensor data.
As mentioned in \cite{HW18}, synchronized aggregate signature schemes are also useful for blockchain protocols.
For instance, we consider a blockchain protocol that records several signed transactions in each new block creation.
The creation of an additional block is a natural synchronization event.
These signed transactions could use a synchronized aggregate signature scheme with a block number as a time period number.
This reduces the signature overhead from one per transaction to just one synchronized signature per block iteration.

\paragraph{\bf{Provable Secure Synchronized Aggregate Signature Schemes.}}
Several provable secure synchronized aggregate signature schemes with bilinear groups have been proposed (see Fig.~\ref{PBSinst}). 

Ahn, Green, and Hohenberger \cite{AGH10} constructed two synchronized aggregate signature schemes based on the Hohenberger-Waters \cite{HW09} short signature scheme.
One is constructed in the random oracle model and the other is constructed in the standard model. 
The security of both schemes relies on the computational Diffie-Hellman (CDH) assumption.

Lee, Lee, and Yung \cite{LLY13} proposed a synchronized aggregate signature scheme based on the Camenisch-Lysyanskaya signature $(\CL)$ scheme \cite{CL04}.
This is the most efficient synchronized aggregate signature scheme with bilinear groups in that the number of pairing operations in the verification of an aggregate signature and the number of group elements in an aggregate signature is smaller than those of \cite{AGH10,GR06}.
The security of this scheme relies on the one-time Lysyanskaya-Rivest-Sahai-Wolf (OT-LRSW) assumption \cite{LRSW99} in the random oracle model. 

As the provable secure synchronized aggregate signature schemes without bilinear groups, Hohenberger and Waters \cite{HW18} proposed the synchronized aggregate signature scheme based on the RSA assumption.

\begin{figure}[t]
\begin{center}
\begin{tabular}{lllcccc}\hline
Scheme &Assumption & Security & $pp$ & ~$\vk$~ & Agg & Agg Ver \\
 & & & size & size &size & (in Pairings) \\
\hline
\hline
GR \cite{GR06} &CDH + ROM &$\rmEUFCMA$ &~$O(1)~$ &~ID~ &3 &3\\
AGH \cite{AGH10} \S 4 &CDH &$\rmEUFCMA$ in CK &$O(k)$ &1 &3 &$k+3$\\
AGH \cite{AGH10} \S A &CDH + ROM &$\rmEUFCMA$ in CK &$O(1)$ &1 &3 &4\\
LLY \cite{LLY13}&OT-LRSW + ROM &$\rmEUFCMA$ in CK &$O(1)$ &1 &2 &3\\
&(interactive assumption) & & & & &\\
LLY \cite{LLY13} &$\rmoMSDHii$ + ROM &$\rmEUFCMA$ in CK &$O(1)$ &1 &2 &3\\
(New proof) & (static assumption)&&&&&\\
\hline
\end{tabular}\\
\begin{flushleft}
In our work, we prove that the scheme LLY \cite{LLY13} satisfies the $\rmEUFCMA$ security in the certified-key model under the $\rmoMSDHii$ assumption in the random oracle model.
\end{flushleft}
\caption{\small Summary of synchronized aggregate signature schemes with bilinear groups.
In the column of $``$Assumption$"$, $``$ROM$"$ means the random oracle model.
In the column of $``$Security$"$, $``$CK$"$ means the certified-key model.
$``$pp size$"$, $``\vk$ size$"$, $``$Agg size$"$, $``$Agg Ver$"$ mean the number of group elements in a public parameter pp, a verification key $\vk$, an aggregate signature, and the number of pairing operations in aggregate signatures verification respectively.
The scheme GR \cite{GR06} is an identity-based scheme that has a verification key size of $``$ID$"$.
In the scheme AGH~\cite{AGH10}, $k$ is a special security parameter.
As mentioned in \cite{AGH10}, $k$ could be five in practice.\protect\linebreak
}
\label{PBSinst}
\end{center}
\end{figure}

\paragraph{\bf{Camenisch-Lysyanskaya Signature Scheme.}}
Camenisch and Lysyanskaya \cite{CL04} proposed the $\CL$ scheme which has a useful feature called randomizability. 
This property allows anyone to randomize a valid signature $\sigma$ to $\sigma'$ where $\sigma$ and $\sigma'$ are valid signatures on the same message.
The $\CL$ scheme is widely used to construct various schemes: anonymous credentials \cite{CL04}, anonymous attestation \cite{BFGSW13}, divisible E-cash \cite{CPST15}, batch verification \cite{CHP07}, group signatures \cite{BCNSW10}, ring signatures \cite{BKM06}, and aggregate signatures \cite{Sch11}.

However, the security of the $\CL$ scheme relies on the Lysyanskaya-Rivest-Sahai-Wolf (LRSW) assumption which is an interactive assumption.
An interactive assumption allows us to design an efficient scheme, however, these are not preferable.

\paragraph{\bf{Modified Camenisch-Lysyanskaya Signature Scheme.}}
Pointcheval and Sanders \cite{PS18} proposed the Modified $q$-Strong Diffie-Hellman-2 $(\rmqMSDHii)$ assumption which is defined on a type 1 bilinear group.
This assumption is a $q$-type assumption \cite{BB04} where the number of input elements depends on the number of adversarial queries. 
They proved that the $\rmqMSDHii$ assumption holds in the generic bilinear group model \cite{BBG05} and the $\CL$ scheme satisfies the weak-existentially unforgeable under chosen message attacks (weak-$\rmEUFCMA$) security under the $\rmqMSDHii$ assumption.

Moreover, they proposed the modified Camenisch-Lysyanskaya signature $(\MCL)$ scheme which has randomizability.
Then, they showed that the $\MCL$ scheme satisfies the existentially unforgeable under chosen message attacks ($\rmEUFCMA$) security under the $\rmqMSDHii$ assumption.
Their modification from the $\CL$ scheme to the $\MCL$ scheme incurs a slight increase in the complexity.\footnote{Their modification from the $\CL$ scheme to the $\MCL$ scheme increases the number of group elements in a signature and an aggregate signature from $2$ to $3$.}

\subsection{Our Results}
To our knowledge, the most efficient synchronized aggregate signature scheme with bilinear groups is Lee et al.'s \cite{LLY13} scheme.
However, the security of this scheme relies on the interactive assumption (the OT-LRSW assumption).
Even if interactive assumptions hold in the generic group model or bilinear group model, the concerns about these assumptions arise in a cryptographic community.
This fact causes a barrier to the use of this scheme.

Also, it is not desired that the security of the scheme depends on $q$-type assumptions.
Because the size of these assumptions grows dynamically and this fact leads to inefficiency of the scheme.
Hence, it is desirable to prove the security of this scheme under the non-$q$-type (static) assumptions or construct another efficient synchronized aggregate signature scheme whose security does not rely on interactive assumptions or $q$-type assumptions.

\paragraph{\bf{Security Proof under the Static Assumption.}}
In this paper, we give a new security proof for Lee et al.'s synchronized aggregate scheme under the static assumption in the random oracle model.
More specifically, we convert from the $\MCL$ scheme to Lee et al.'s \cite{LLY13} synchronized aggregate signature scheme.
Then, we reduce the security of Lee et al.'s scheme to the one-time $\rmEUFCMA$ $(\rmOTEUFCMA)$ security of the $\MCL$ scheme in the random oracle model.
We refer the reader to Section \ref{LLMSection} for details about these techniques.

Since the $\rmOTEUFCMA$ security of the $\MCL$ scheme is implied by the $\rmoMSDHii$ assumption, the security of Lee et al.'s scheme can be proven under the $\rmoMSDHii$ assumption.
We can regard the $\rmoMSDHii$ assumption as the static assumption.
Therefore, we can see that the security of Lee et al.'s scheme relies on the static assumption.
Notably, while the $\rmEUFCMA$ security of the $\MCL$ scheme is proved under the $q$-type assumption, the security of Lee et al.'s synchronized aggregate signature scheme can be proven under the static assumption in the random oracle model.

\paragraph{\bf{Trade-offs with Little Loss of Efficiency in the Reduction.}}
In general, there is a trade-off that efficiency is reduced when we design a scheme based on weaker computational assumptions.
Surprisingly, we can change the assumptions underlying the security of Lee et al.'s \cite{LLY13} scheme from the interactive assumption (OT-LRSW) to the static assumption ($\rmoMSDHii$) with little loss in the efficiency of the scheme.
Specifically, the size of verification key $\vk$, the size of aggregate signature $\Sigma$, and the number of pairing operations in an aggregate signature verification do not increase at all.

\subsection{Related Works}
Boneh et. al 's \cite{BGLS03} proposed the first full aggregate signature scheme which allows any user to aggregate signatures of different signers.
Furthermore, this scheme allows us to aggregate individual signatures as well as already aggregated signatures in any order. 
They constructed a full aggregate signature scheme in the random oracle model.
Hohenberger, Sahai, and Waters \cite{HSW13} firstly constructed a full aggregate signature scheme in the standard model by using multilinear maps.
Hohenberger, Koppula, and Waters \cite{HKW15} constructed a full aggregate signature scheme in the standard model by using the indistinguishability obfuscation.

Several variants of aggregate signature schemes have been proposed.
One major variant is a sequential aggregate signature scheme which was firstly proposed by Lysyanskaya, Micali, Reyzin, and Shacham \cite{LMRS04}.
In this scheme, an aggregate signature is constructed sequentially, with each signer modifying the aggregate signature in turn.
They constructed a sequential aggregate signature scheme in the random oracle model by using families of trapdoor permutations.
Lu, Rafail Ostrovsky, Sahai, Shacham, and Waters \cite{LOSSW06} firstly constructed the sequential aggregate signature scheme in the standard model based on the Waters signature scheme.
Another major variant of aggregate signature schemes is a synchronized aggregate signature scheme explained in Section \ref{IntroBack}.
Furthermore, Lee et. al \cite{LLY13} proposed a combined aggregate signature scheme.
In this scheme, a signer can use two modes of aggregation (sequential aggregation or synchronized aggregation) dynamically.
They constructed a combined aggregate signature scheme in the random oracle model based on the $\CL$ scheme.

\subsection{Road Map}
In Section \ref{Prelimi}, we recall bilinear groups, the $\rmoMSDHii$ assumption, and a digital signature scheme.
In Section \ref{SyncAggSig}, we review the definition of a synchronized aggregate signature scheme and its security notion.
In Section \ref{LLMSection}, we review the $\MCL$ scheme.
Next, we explain the relationship between the $\MCL$ scheme and Lee et al.'s aggregate signature scheme.
In particular, we explain how to convert from the $\MCL$ to Lee et al.'s aggregate signature scheme.
Then, we describe Lee et al.'s aggregate signature scheme construction and newly give a security proof under the $\rmoMSDHii$ assumption in the random oracle model.

\section{Preliminaries}\label{Prelimi}

Let $1^{\lambda}$ be the security parameter. 
A function $f(\lambda)$ is negligible in $\lambda$ if $f(\lambda)$ tends to $0$ faster than $\frac{1}{\lambda^c}$ for every constant $c > 0$.
PPT stands for probabilistic polynomial time.
For an integer $n$, $[n]$ denotes the set $\{1,\dots, n\}$.
For a finite set $S$, $s \xleftarrow{\$} S$ denotes choosing an element $s$ from $S$ uniformly at random. 
For a group $\G$, we define $\G^* := \G \backslash \{1_{\G}\}$. 
For an algorithm $\A$, 
$y \leftarrow \A(x)$ denotes that the algorithm $\A$ outputs $y$ on input $x$.

\subsection{Bilinear Groups}
In this work, we use type $1$ pairings and introduce a bilinear group generator.
Let $\mathsf{G}$ be a bilinear group generator that takes as an input a security parameter $1^{\lambda}$ and outputs the descriptions of multiplicative groups $\mathcal{G} = (p, \G, \G_T, e)$ where $\G$ and $\G_T$ are groups of prime order $p$ and $e$ is an efficient computable, non-degenerating bilinear map $e:\G \times \G \rightarrow \G_T$.
\begin{enumerate}
\item Bilinear: for all $u \in \G$, $v \in \G$ and $a, b \in \Z_{p}$, then $e(u^a, v^b) = e(u, v)^{ab}$.
\item Non-degenerate: for any $g \in \G^*$ and $\tilde{g} \in \G^*$, $e(g, \tilde{g}) \neq 1_{\G_T}$.
\end{enumerate}

\subsection{Computational Assumption}
Pointcheval and Sanders \cite{PS18} introduced the new $q$-type assumption which is called the Modified $q$-Strong Diffie-Hellman-2 $(\rmqMSDHii)$ assumption.
This is a variant of the $q$-Strong Diffie-Hellman ($q$-SDH) assumption and defined on a type 1 bilinear group.
The $\rmqMSDHii$ assumption holds in the generic bilinear group model \cite{BBG05}.
In this work, we fix the value to $q = 1$ and only use $\rmoMSDHii$ assumption in a static way.
We can regard $\rmoMSDHii$ as a static assumption.

\begin{assumption}[Modified 1-Strong Diffie-Hellman-2 Assumption \cite{PS18}]
Let $\mathsf{G}$ be a type-1 pairing-group generator.
The Modified $1$-Strong Diffie-Hellman-2 $(\rmoMSDHii)$ assumption over $\mathsf{G}$ is that for all $\lambda \in \mathbb{N}$, for all $\mathcal{G}= (p, \G, \G_T, e) \leftarrow \mathsf{G}(1^\lambda)$, given $(\mathcal{G}, g, g^{x},g^{x^{2}}, g^b, g^{bx}, g^{bx^{2}}, g^a, g^{abx})$ where $g \leftarrow \G^*$ and $a,b,x \xleftarrow{\$} \mathbb{Z}_p^*$ as an input, no PPT adversary can, without non-negligible probability, output a tuple $(w, P, h^{\frac{1}{x+w}}, h^{\frac{a}{x \cdot P(x)}})$ with $h \in \G$, $P$ a polynomial in $\mathbb{Z}_p[X]$ of degree at most $1$, and $w \in \mathbb{Z}_p^*$ such that $X+w$ and $P(X)$ are relatively prime.\footnote{In the $\rmqMSDHii$ assumption, an input is changed to $(\mathcal{G}, g, g^{x},\dots, g^{x^{q+1}},\allowbreak g^b, g^{bx},\allowbreak \dots, g^{bx^{q+1}}, g^a, g^{abx})$ and the condition of the order of $P(x)$ is changed to at most $q$.}
\end{assumption}

\subsection{Digital Signature Schemes}
We review the definition of a digital signature scheme and its security notion.
\begin{definition}[Digital Signature Scheme]
A digital signature scheme $\DS$ consists of following four algorithms $(\Setup, \KeyGen, \allowbreak \Sign, \Verify)$.
\begin{itemize}
\item $\Setup (1^{\lambda}):$ Given a security parameter $\lambda$, return the public parameter pp.
We assume that $pp$ defines the message space $\mathcal{M}_{pp}$.

\item $\KeyGen (pp):$ Given a public parameter $pp$, return a verification key $\vk$ and a signing key $\sk$.

\item $\Sign (pp, \sk, m):$ Given a public parameter $pp$, a signing key $\sk$, and a message $m \in \mathcal{M}_{pp}$, return a signature $\sigma$.

\item $\Verify (pp, \vk, m, \sigma):$ Given a public parameter $pp$, a verification key $\vk$, a message $m \in \mathcal{M}_{pp}$, and a signature $\sigma$, return either $1$ (Accept) or $0$ (Reject).
\end{itemize}
Correctness: 
Correctness is satisfied if for all $\lambda \in \N$, $pp \leftarrow \Setup (1^{\lambda})$ for all $m \in \mathcal{M}_{pp}$, $(\vk, \sk) \leftarrow \KeyGen(pp)$, and $\sigma \leftarrow \Sign(pp, \sk, m)$, $\Verify(pp, \vk, m, \sigma) = 1$ holds.
\end{definition}

The $\rmEUFCMA$ security \cite{GMR88} is the standard security notion for digital signature schemes.

\begin{definition}[EUF-CMA Security \cite{GMR88}]
The $\rmEUFCMA$ security of a digital signature scheme $\DS$ is defined by the following unforgeability game between a challenger $\C$ and a PPT adversary $\A$.

\begin{itemize}
\item $\C$ runs $pp \leftarrow \Setup(1^{\lambda})$, $(\vk, \sk) \leftarrow \KeyGen(pp)$, sets $Q \leftarrow \{\}$, and gives $(pp, \vk)$ to $\A$.
\item $\A$ is given access (throughout the entire game) to a sign oracle $\mathcal{O}^{\Sign}(\cdot)$.
Given an input $m$, $\mathcal{O}^{\Sign}$ sets $Q \leftarrow Q \cup \{m\}$ and returns $\sigma \leftarrow \Sign(pp, \sk, m)$.
\item $\A$ outputs a forgery $(m^*, \sigma^*)$.
\end{itemize}
A digital signature scheme $\DS$ satisfies the $\rmEUFCMA$ security if for all PPT adversaries $\A$, the following advantage
\begin{equation*}
\Adv^{\EUFCMA}_{\DS, \A}:=\Pr[\Verify(pp, \vk, m^*, \sigma^*) = 1 \land m^* \notin Q]
\end{equation*}
is negligible in $\lambda$.
\end{definition}
If the number of signing oracle $\mathcal{O}^{\Sign}$ query is restricted to the one-time in the unforgeability security game, we call $\DS$ satisfies the one-time $\rmEUFCMA$ ($\rmOTEUFCMA$) security.

\section{Synchronized Aggregate Signature Schemes}\label{SyncAggSig}
In this section, we review the definition of a synchronized aggregate signature scheme and its security notion.

\subsection{Synchronized Aggregate Signature Schemes}
Synchronized aggregate signature schemes \cite{AGH10,GR06} are a special type of aggregate signature schemes.
In this scheme, all of the signers have a synchronized time period $t$ and each signer can sign a message at most once for each period~$t$.
A set of signatures that are all generated for the same period $t$ can be aggregated into a short signature.
The size of an aggregate signature is the same size as an individual signature.
Now, we review the definition of synchronized aggregate signature schemes.

\begin{definition} [Synchronized Aggregate Signature Schemes \cite{AGH10,GR06}]\label{SyncAggDef}
A synchronized aggregate signature scheme $\SyncAS$ for a bounded number of periods is a tuple of algorithms $(\SyncASSetup, \allowbreak \SyncASKeyGen,\allowbreak \SyncASSign, \SyncASVerify,\allowbreak \SyncASAggregate, \allowbreak \SyncASAggVerify)$.

\begin{itemize}
\item $\SyncASSetup(1^\lambda, 1^T): $ Given a security parameter $\lambda$ and the time period bound $T$, return the public parameter pp.
We assume that $pp$ defines the message space $\mathcal{M}_{pp}$.

\item $\SyncASKeyGen (pp):$ Given a public parameter $pp$, return a verification key $\vk$ and a signing key $\sk$.

\item $\SyncASSign (pp, \sk, t, m):$ Given a public parameter $pp$, a signing key $\sk$, a time period $t \leq T$, and a message $m \in \mathcal{M}_{pp}$, return the signature $\sigma$.

\item $\SyncASVerify (pp, \vk, m, \sigma):$ Given a public parameter $pp$, a verification key $\vk$, a message $m \in \mathcal{M}_{pp}$, and a signature $\sigma$, return either $1$ (Accept) or $0$ (Reject).

\item $\SyncASAggregate (pp, (\vk_1,\dots, \vk_r), (m_1, \dots, m_r), \allowbreak (\sigma_1, \dots, \sigma_r)):$
Given a public parameter $pp$, a list of verification keys $(\vk_1,\dots, \vk_r)$, a list of messages $(m_1, \dots, m_r)$, and a list of signatures $(\sigma_1, \dots, \sigma_r)$, return either the aggregate signature $\Sigma$ or~$\bot$.  
\item $\SyncASAggVerify (pp, (\vk_1,\dots, \vk_r), (m_1, \dots, m_r), \Sigma):$
Given a public parameter $pp$, a list of verification keys $(\vk_1,\dots, \vk_r)$, a list of messages $(m_1, \dots, m_r)$, and an aggregate signature, return either $1$ (Accept) or $0$ (Reject).
\end{itemize}
Correctness:
Correctness is satisfied if for all $\lambda \in \mathbb{N}$, $T \in \mathbb{N}$, $pp \leftarrow \SyncASSetup(1^\lambda, 1^{T})$, for any finite sequence of key pairs $(\vk_1, \sk_1),\dots (\vk_r, \sk_r) \leftarrow \SyncASKeyGen (pp)$ where $\vk_i$ are all distinct, for any time period $t \leq T$, for any sequence of messages $(m_1, \dots m_r) \in \mathcal{M}_{pp}$, 
$\sigma_i \leftarrow \SyncASSign (pp, \sk_i, t, m_i)$ for $i \in [r]$, 
$\Sigma \leftarrow \SyncASAggregate(pp, (\vk_1,\dots, \vk_r), (m_1, \dots, m_r), (\sigma_1, \dots, \sigma_r))$, we have
\begin{equation*}
\begin{split}
\SyncASVerify &(pp, \vk_i, m_i, \sigma_i) = 1 {\rm \ for \ all \ } i \in [r]\\
&\land \SyncASAggVerify (pp, (\vk_1,\dots, \vk_r), (m_1, \dots, m_r), \Sigma) = 1.
\end{split}
\end{equation*}
\end{definition}

In a signature aggregation, it is desirable to confirm that each signature is valid.
This is because if there is at least one invalid signature, the generated aggregate signature will be invalid.\footnote{Fault-tolerant aggregate signature schemes \cite{HKKKR16} allow us to determine the subset of all messages belonging to an aggregate signature that were signed correctly.
However, this scheme has a drawback that the aggregate signature size depends on the number of signatures to be aggregated into it.}
In this work, before aggregating signatures, $\SyncASAggregate$ checks the validity of each signature.

\subsection{Security of Synchronized Aggregate Signature Schemes}
We introduce the security notion of synchronized aggregate signature schemes.
The $\rmEUFCMA$ security of synchronized aggregate signature schemes proposed by Gentry and Ramzan \cite{GR06} captures that it is hard for adversaries to forge an aggregate signature without signing key $\sk^*$.
However, they only provided heuristic security arguments in their synchronized aggregate signature scheme.

Ahn, Green, and Hohrnberger \cite{AGH10} introduced the certified-key model for the $\rmEUFCMA$ security of synchronized aggregate signature schemes.
In this model, signers must certify their verification key $\vk$ by proving knowledge of their signing key $\sk$.
In other words, no verification key $\vk$ is allowed except those correctly generated by the $\SyncASKeyGen$ algorithm.
In certified-key model, to ensure the correct generation of a verification key $\vk_i \neq \vk^*$, $\rmEUFCMA$ adversaries must submit $(\vk_i, \sk_i)$ to the certification oracle $\mathcal{O}^{\Cert}$.
As in \cite{AGH10,LLY13}, we consider the $\rmEUFCMA$ security in the certified-key model.
\begin{definition}[EUF-CMA Security in the Certified-Key Model \cite{AGH10,LLY13}]
The $\rmEUFCMA$ security of a sequential aggregate signature scheme $\SyncAS$ in the certified-key model is defined by the following unforgeability game between a challenger $\C$ and a PPT adversary $\A$.

\begin{itemize}
\item $\C$ runs $pp^* \leftarrow \SyncASSetup(1^\lambda, 1^{T})$, $(\vk^*, \sk^*) \leftarrow \SyncASKeyGen(pp^*)$, sets $Q \leftarrow \{\}$, $L \leftarrow \{\}$, $t_{ctr} \leftarrow 1$, and gives $(pp, \vk^*)$ to $\A$.

\item $\A$ is given access (throughout the entire game) to a certification oracle $\mathcal{O}^{\Cert}(\cdot, \cdot)$.
Given an input $(\vk, \sk)$, $\mathcal{O}^{\Cert}$ performs the following procedure.
\begin{itemize}
\item If the key pair $(\vk, \sk)$ is valid, $L \leftarrow L \cup \{\vk\}$ and return $``\accept"$.
\item Otherwise return $``\reject"$.
\end{itemize}
$($$\A$ must submit key pair $(\vk, \sk)$ to $\mathcal{O}^{\Cert}$ and get $``\accept"$ before using $\vk$.$)$

\item $\A$ is given access (throughout the entire game) to a sign oracle $\mathcal{O}^{\Sign}(\cdot, \cdot)$.
Given an input $(``\Inst", m)$, $\mathcal{O}^{\Sign}$ performs the following procedure.\\
$($$``\Inst" \in \{``\instskip", ``\instsign"\}$ represent the instruction for $\mathcal{O}^{\Sign}$ where $``\instskip"$ implies that $\A$ skips the concurrent period $t_{ctr}$ and $``\instsign"$ implies that $\A$ require the signature on message $m$. $)$
\begin{itemize}
\item If $t_{ctr} \notin [T]$, return $\bot$.
\item If $``\Inst" = ``\instskip"$, $t_{ctr} \leftarrow t_{ctr} +1$.
\item If $``\Inst" = ``\instsign"$, $Q \leftarrow Q \cup \{m\}$, $\sigma \leftarrow\SyncASSign (pp^*, \sk^*, t, m)$, $t_{ctr} \leftarrow t_{ctr} +1$, return~$\sigma$.
\end{itemize}
\item $\A$ outputs a forgery $((\vk^*_1,\dots, \vk_{r^*}^*), (m_1^*, \dots, m_{r^*}^*), \Sigma^*)$.
\end{itemize}
A sequential aggregate signature scheme $\SyncAS$ satisfies the $\rmEUFCMA$ security in the certified-key model if for all PPT adversaries $\A$, the following advantage
\begin{equation*}
\Adv^{\EUFCMA}_{\SyncAS, \A}:=\Pr\left[
\begin{split}
&\SyncASAggVerify(pp^*, (\vk^*_1,\dots, \vk_{r^*}^*), (m_1^*, \dots, m^*_{r^*}), \Sigma^*) = 1\\
&\land {\rm For \ all \ } j \in [r^*] {\rm \ such \ that \ } \vk^*_j \neq \vk^*, \vk^*_j \in L \\
&\land {\rm For \ some \ } j^* \in [r^*] {\rm \ such \ that \ } \vk^*_{j^*} = \vk^*, m^*_{j^*} \notin Q \\
\end{split}
\right]
\end{equation*}
is negligible in $\lambda$.

\end{definition}

\section{Lee et al.'s Aggregate Signature Scheme}\label{LLMSection}
In this section, first, we review the $\MCL$ scheme proposed by Pointcheval and Sanders \cite{PS18}.
Next, we explain an intuition that there is a relationship between the $\MCL$ scheme and Lee et al.'s aggregate signature scheme.
Concretely, we explain that there is a conversion from the $\MCL$ scheme to Lee et al.'s aggregate signature scheme. 
Then, we describe Lee et al.'s aggregate signature scheme construction.
Finally, we give a new security proof for Lee et al.'s scheme under the $\rmoMSDHii$ assumption in the random oracle model.

\subsection{Modified Camenisch-Lysyanskaya Signature Scheme}\label{MCLSS}

Pointcheval and Sanders \cite{PS18} proposed the modified Camenisch-Lysyanskaya signature scheme which supports a multi-message (vector message) signing.
In this work, we only need a single-message signing scheme.
Here, we review the single-message modified Camenisch-Lysyanskaya signature scheme $\MCL = (\MCLSetup, \allowbreak \MCLKeyGen, \MCLSign, \MCLVerify)$ as follows.
\begin{itemize}
\item $\MCLSetup(1^\lambda):$\\
~~~$\mathcal{G}=(p, \G, \G_T, e) \leftarrow \mathsf{G}(1^\lambda)$.\\
~~~Return $pp \leftarrow \mathcal{G}$.

\item $\MCLKeyGen(pp):$\\
~~~$g \xleftarrow{\$} \mathbb{G}^*$, $x \xleftarrow{\$} \mathbb{Z}_p^*$, $y\xleftarrow{\$} \mathbb{Z}_p^*$, $z \xleftarrow{\$} \mathbb{Z}_p^*$, $X \leftarrow g^{x}$, $Y \leftarrow g^{y}$, $Z \leftarrow g^{z}$.\\
~~~Return $(\vk, \sk) \leftarrow ((g, X, Y, Z), (x, y, z))$.

\item$\MCLSign (pp, \sk, m):$\\
~~~Parse $\sk$ as $(x, y, z)$\\
~~~$w \xleftarrow{\$} \mathbb{Z}_p$, $A \xleftarrow{\$} \mathbb{G}^*$, $B \leftarrow A^y$, $C \leftarrow A^z$, $D \leftarrow C^y$, $E \leftarrow A^xB^{mx}D^{wx}$.\\
~~~Return $\sigma \leftarrow (w, A, B, C, D, E)$.

\item $\MCLVerify (pp, \vk, m, \sigma):$\\
~~~Parse $\vk$ as $(g, X, Y, Z)$, $\sigma$ as $(w, A, B, C, D, E)$.\\
~~~If $\left(e(A, Y) \neq e(B, g)\right) \lor \left(e(A, Z) \neq e(C, g)\right) \lor \left(e(C, Y) \neq e(D, g)\right)$, return $0$. \\
~~~If $e(AB^mD^w, \tilde{X}) = e (E,g)$, return $1$.\\
~~~Otherwise return $0$.\\
\end{itemize}

Pointcheval and Sanders \cite{PS18} proved that if the $\rmqMSDHii$ assumption holds, then the $\MCL$ scheme satisfies the $\rmEUFCMA$ security where $q$ is a bound on the number of adaptive signing queries.
In this work, we only need the $\rmOTEUFCMA$ security for the $\MCL$ scheme.
\begin{theorem}[\cite{PS18}]\label{MCLtoqSDHii}
If the $\rmoMSDHii$ assumption holds, then the $\MCL$ scheme satisfies the $\rmOTEUFCMA$ security.
\end{theorem}

\subsection{Conversion to Lee et al.'s Aggregate Signature Scheme}\label{HighLevelIdea}

We explain that the $\MCL$ scheme can be converted into Lee et al.'s aggregate signature scheme.
Our idea of conversion is a similar technique in \cite{LLY13} which converts the Camenisch-Lysyanskaya signature $\CL$ scheme to the synchronized aggregate signature scheme.

Now, we explain an intuition of our conversion.
We start from the $\MCL$ scheme in Section \ref{MCLSS}.
A signature of the $\MCL$ scheme on a message $m$ is formed as
\begin{equation*}
\sigma = (w, A, B=A^y, C=A^z, D=C^y, E=A^xB^{mx}D^{wx}).
\end{equation*}
where $w \xleftarrow{\$} \mathbb{Z}_p$ and $A \xleftarrow{\$} \mathbb{G}_1^*$.
If we can force signers to use same $w$, $A$, $B=A^y$, $C=A^z$, and $D=C^y$, we can obtain an aggregate signature 
\begin{equation*}
\Sigma = \left(w, A, B, C, D, E'= \prod^r_{i=1}E_i =A^{\sum^{r}_{i=1}x_i}B^{\sum^{r}_{i=1}m_ix_i}D^{\sum^{r}_{i=1}wx_i}\right)
\end{equation*}
on a message list $(m_1,\dots, m_r)$ from valid signatures $(\sigma_1,\dots\sigma_r )$ where $\sigma_i = (w_, A, B, \allowbreak C, D, E_i)$ is a signature on a message $m_i$ generated by each signer.
If we regard $E'$ as $E'=(AD^w)^{\sum^{r}_{i=1}x_i} B^{\sum^{r}_{i=1}m_ix_i}$, verification of the aggregate signature $\Sigma$ on the message list $(m_1,\dots, m_r)$ can be done by checking the following equation. 
\begin{equation*}
e(E', g) =e\left(AD^w , \prod^{r}_{i=1} \vk_i \right) \cdot e\left(B , \prod^{r}_{i=1} \vk_i^{m_i} \right)
\end{equation*}
Then, required elements to verify the aggregate signature $\Sigma$ are $F=AD^w$, $B$, and $E'$.
Similar to Lee et al.'s conversion, the three verification equations $e(A, Y) = e(B, g)$, $e(A, Z) = e(C, g)$, $e(C, Y) = e(D, g)$ in $\MCLVerify$ is discarded in this conversion.
This does not affect the security proof in Section \ref{LLMSyncProof}.
We use hash functions to force signers to use the same $F$ and $B$ for each period~$t$.
We choose hash functions $H_1$ and $H_2$ and set $F \leftarrow H_1(t)$ and $B \leftarrow H_2(t)$.
Then, we can derive Lee et al.'s aggregate signature scheme.
In this derived aggregate signature scheme, a signature on a message $m$ and period $t$ is formed as
\begin{equation*}
\sigma = (E=H_1(t)^{x}H_2(t)^{mx}, t).
\end{equation*}
An aggregate signature $\Sigma'$ on a message list $(m_1,\dots, m_r)$ and period $t$ is formed as
\begin{equation*}
\Sigma = \left(E'= \prod^r_{i=1}E_i = H_1(t)^{\sum^{r}_{i=1}x_i}H_2(t)^{\sum^{r}_{i=1}m_ix_i}, t\right)
\end{equation*}
where $\sigma_i = (E_i=H_1(t)^{x_i}H_2(t)^{m_ix_i}, t)$ is a signature on a message $m_i$ generated by each signer.
In our conversion, we need to hash a message with a time period for the security proof.
This conversion is used for the reduction algorithm $\B$ in Section \ref{LLMSyncProof}.

\subsection{Lee et al.'s Synchronized Aggregate Signature Scheme}
We describe Lee et al.'s synchronized aggregate signature scheme obtained by adapting the conversion in Section \ref{HighLevelIdea} to the $\MCL$ scheme.
Let $T$ be a bounded number of periods which is a polynomial in $\lambda$.
The Lee et al.'s synchronized aggregate signature scheme $\MCLSyncAS = (\MCLSyncASSetup, \allowbreak \MCLSyncASKeyGen, \allowbreak \MCLSyncASSign, \allowbreak \MCLSyncASVerify, \allowbreak \MCLSyncASAggregate, \MCLSyncASAggVerify)$ \cite{LLY13} is given as follows.\footnote{The $\MCLSyncAS$ scheme described here is slightly different from the original ones \cite{LLY13} in that the range of $H_2$ is changed from $\G$ to $\G^*$.}

\begin{itemize}
\item $\MCLSyncASSetup(1^\lambda, 1^{T}):$
\begin{enumerate}
\item $\mathcal{G}=(p, \G, \G_T, e) \leftarrow \mathsf{G}(1^\lambda)$, $g \xleftarrow{\$} \G^*$.
\item Choose hash functions: \\
$H_1: [T] \rightarrow \G$, $H_2: [T] \rightarrow \G^*$, $H_3: [T] \times \{0, 1\}^* \rightarrow \mathbb{Z}_p$. 
\item Return $pp \leftarrow (\mathcal{G}, g, H_1, H_2, H_3)$.
\end{enumerate}

\item $\MCLSyncASKeyGen (pp):$
\begin{enumerate}
\item $x \xleftarrow{\$} \mathbb{Z}^*_p$, $X \leftarrow g^{x}$.
\item Return $(\vk, \sk) \leftarrow (X, x)$.
\end{enumerate}

\item $\MCLSyncASSign (pp, \sk, t, m):$
\begin{enumerate}
\item $m' \leftarrow H_3(t,m)$, $E \leftarrow H_1(t)^{\sk}H_2(t)^{m'\sk}$.
\item Return $(E, t)$.
\end{enumerate}

\item $\MCLSyncASVerify (pp, \vk, m, \sigma):$
\begin{enumerate}
\item $m' \leftarrow H_3(t,m)$, parse $\sigma$ as $(E, t),$.
\item If $e(E,g) = e(H_1(t)H_2(t)^{m'}, \vk)$, return $1$.
\item Otherwise return $0$.
\end{enumerate}

\item $\MCLSyncASAggregate (pp, (\vk_1,\dots, \vk_r), (m_1, \dots, m_r), \allowbreak (\sigma_1, \dots, \sigma_r)):$
\begin{enumerate}
\item For $i=1$ to $r$, parse $\sigma_i$ as $(E_i, t_i)$.
\item If there exists $i \in \{2,\dots, r\}$ such that $t_i \neq t_1$, return $\bot.$
\item If there exists $(i, j) \in [r]\times [r]$ such that $i \neq j \land \vk_i = \vk_j$, return $\bot$.
\item If there exists $i \in [r]$ suth that $\MCLSyncASVerify (pp, \vk_i, m_i, \sigma_i) \neq 0$,\\ ~~ return~$\bot$.
\item $E' \leftarrow \prod^{r}_{i=1} E_i$.
\item Return $\Sigma \leftarrow (E', w)$.
\end{enumerate}

\item $\MCLSyncASAggVerify (pp, (\vk_1,\dots, \vk_r), (m_1, \dots, m_r), \Sigma):$
\begin{enumerate}
\item There exists $(i, j) \in [r]\times [r]$ such that $i \neq j \land \vk_i = \vk_j$, return $0$.
\item For $i= 1$ to $r$, $m'_i \leftarrow H_3(t,m_i)$.
\item Parse $\Sigma$ as $(E', w)$.
\item If $e(E',g) =e\left(H_1(t) , \prod^{r}_{i=1} \vk_i \right) \cdot e\left(H_2(t) , \prod^{r}_{i=1} \vk_i^{m'_i} \right)$, return $1$.
\item Otherwise, return $0$.
\end{enumerate}

\end{itemize}

Now, we confirm the correctness.
Let $(\vk_i, \sk_i) \leftarrow \MCLSyncASKeyGen (pp)$ and $\sigma_i \leftarrow \MCLSyncASSign (pp, \allowbreak \sk_i, t, m_i)$ for $i \in [r]$ where $\vk_i$ are all distinct.
Then, for all $i \in [r]$, $E_i \leftarrow H_1(t)^{\sk_i}H_2(t)^{m'_i\sk_i}$ holds where $m'_i \leftarrow H_3(t, m_i)$ and $\sigma_i = (E_i,t)$.
This fact implies that $\MCLSyncASVerify (pp, \vk_i, m_i, \sigma_i)=1$. Furthermore, let $\Sigma \leftarrow \MCLSyncASAggregate (pp, (\vk_1,\dots, \vk_r), (m_1, \dots, m_r), \allowbreak (\sigma_1, \dots, \sigma_r))$.
Then, 
\begin{equation*}
E' = \prod^{r}_{i=1}E_i = H_1(t)^{\sum^n_{i=1} \sk_i} H_2(t)^{\sum^n_{i=1} m'_i\sk_i}
\end{equation*}
holds where $\Sigma = (E', t)$ and $m'_i \leftarrow H_3(t, m_i)$ for all $i \in [r]$.
This fact implies that $\MCLSyncASAggVerify (pp,\allowbreak (\vk_1,\dots, \vk_r), (m_1, \dots, m_r), \Sigma)=1$.

\subsection{New Security Proof under the Static Assumption}\label{LLMSyncProof}
We reassess the $\rmEUFCMA$ security of the $\MCLSyncAS$ scheme.
In particular, we newly prove the $\rmEUFCMA$ security of the $\MCLSyncAS$ scheme under the $\rmoMSDHii$ assumption.

\begin{theorem}\label{MCLSyncASEUFCMA}
If the $\MCL$ scheme satisfies the $\rmOTEUFCMA$ security, then, in the random oracle model, the $\MCLSyncAS$ scheme satisfies the $\rmEUFCMA$ security in the certified-key model.
\end{theorem}

\begin{proof}
We give an overview of our security proof.
Similar to the work in \cite{LLY13}, we reduce the $\rmEUFCMA$ security of the $\MCLSyncAS$ scheme to the $\rmOTEUFCMA$ security of the $\MCL$ scheme. 
We construct a reduction algorithm according to the following strategy.
First, the reduction algorithm chooses a message $m_\MCL$ at random, make signing query on $m_\MCL$, and obtains its signature $\sigma_{\MCL}=(w_{\MCL}, A_{\MCL}, B_{\MCL}, C_{\MCL}, \allowbreak D_{\MCL}, \allowbreak E_{\MCL})$ of the $\MCL$ scheme.
Then, the reduction algorithm guesses the time period $t'$ of a forged aggregate signature and an index $k' \in [q_{H_3}]$ at random where $q_{H_3}$ be the maximum number of $H_3$ hash queries.
Then reduction algorithm programs hash values as $H_1(t') = A_{\MCL}D_{\MCL}^{w_{\MCL}}$, $H_2(t')=B_{\MCL}$, and $H_3(t',m_{k'}) = m_{\MCL}$.
For a signing query on period $t \neq t'$, the reduction algorithm generate the signature by programmability of hash functions $H_1$, $H_2$, and $H_3$.
For a signing query on period $t \neq t'$, if the query index $j$ of $H_3$ is equal to the index $k'$, the reduction algorithm can compute a valid signature by using $\sigma_{\MCL}$ (This can be done by using the conversion technique in Section \ref{HighLevelIdea}.).
Otherwise, the algorithm should abort the simulation.
Finally, the reduction algorithm extracts valid forgery of the $\MCL$ scheme from a forged aggregate signature on time period $t'$ of the $\MCLSyncAS$ scheme.

Now, we give the security proof. 
Let $\A$ be an $\rmEUFCMA$ adversary of the $\MCLSyncAS$ scheme, $\C$ be the $\rmOTEUFCMA$ game challenger of the $\MCL$ scheme, and $q_{H_3}$ be the maximum number of $H_3$ hash queries.
We construct the algorithm $\B$ against the $\rmOTEUFCMA$ game of the $\MCL$ scheme.
The construction of $\B$ is given as follow.
\begin{itemize}
\item {\bf Initial setup:}
Given an input $pp=\mathcal{G}_{\MCL}$ and $\vk = (g_{\MCL}, X_{\MCL}, Y_{\MCL}, Z_{\MCL})$ from $\C$, 
$\B$ performs the following procedure.
\begin{itemize}
\item $\mathcal{G} \leftarrow \mathcal{G}_{\MCL}$, $g \leftarrow g_{\MCL}$, $pp^* \leftarrow (\mathcal{G}, g)$, $\vk^* \leftarrow X_{\MCL}$. $t' \xleftarrow{\$} [T]$, $k' \xleftarrow{\$} [q_{H_3}]$, $t_{ctr} \leftarrow 1$, $L \leftarrow \{\}$, $K \leftarrow \{\}$, $\mathbb{T}_1 \leftarrow \{\}$, $\mathbb{T}_2 \leftarrow \{\}$, $\mathbb{T}_3 \leftarrow \{\}$, $Q \leftarrow \{\}$. 
\item $m_{\MCL} \xleftarrow{\$} \mathbb{Z}_p$, query $\C$ for the signature on the message $m_{\MCL}$ and get its signature $\sigma_{\MCL} = (w_{\MCL}, A_{\MCL}, B_{\MCL}, C_{\MCL}, D_{\MCL}, \allowbreak E_{\MCL})$, 
\item Send $(pp^*, \vk^*)$ to $\A$ as an input.
\end{itemize}

\item $\mathcal{O}^{\Cert}(\vk, \sk):$
If $\vk = g^{\sk}$, update lists $L \leftarrow L \cup \{\vk\}$, $K \leftarrow K \cup \{(\vk, \sk)\}$ and return $``\accept"$ to $\A$.
Otherwise return $``\reject"$ to $\A$.

\item $\mathcal{O}^{H_1}(t_i):$
Given an input $t_i$, $\B$ responds as follows.
\begin{itemize}
\item If there is an entry $(t_i, \cdot , F_i)$ (`$\cdot$' represents an arbitrary value or $\bot$) for some $F_i \in \G_1$ in $\mathbb{T}_1$, return $F_i$.
\item If $t_i \neq t'$, $r_{(1,i)} \xleftarrow{\$} \mathbb{Z}_p$, $F_i \leftarrow g^{r_{(1,i)}}$, $\mathbb{T}_1 \leftarrow \mathbb{T}_1 \cup \{(t_i, r_{(1,i)} ,F_i)\}$, return $F_i$.
\item If $t_i = t'$, $\mathbb{T}_1 \leftarrow \mathbb{T}_1 \cup \{(t_i, \bot , A_{\MCL}D_{\MCL}^{w_{\MCL}}\}$, return $A_{\MCL}D_{\MCL}^{w_{\MCL}}$.
\end{itemize}

\item $\mathcal{O}^{H_2}(t_i):$
Given an input $t_i$, $\B$ responds as follows.
\begin{itemize}
\item If there is an entry $(t_i, \cdot , B_i)$ (`$\cdot$' represents an arbitrary value or $\bot$) for some $B_i \in \G_1^*$ in $\mathbb{T}_3$, return $B_i$.
\item If $t_i \neq t'$, $r_{(2,i)} \xleftarrow{\$} \mathbb{Z}^*_p$, $B_i \leftarrow g^{r_{(2,i)}}$, $\mathbb{T}_2 \leftarrow \mathbb{T}_2 \cup \{(t_i, r_{(2,i)} ,B_i)\}$, return~$D_i$.
\item If $t_i = t'$, $\mathbb{T}_2 \leftarrow \mathbb{T}_2 \cup \{(t_i, \bot , B_{\MCL})\}$, return $B_{\MCL}$.
\end{itemize}

\item $\mathcal{O}^{H_3}(t_i, m_j):$
Given an input $(t_i, m_j)$, $\B$ responds as follows.
\begin{itemize}
\item If there is an entry $(t_i, m_j, m'_{(i,j)})$ for some $m'_{(i,j)} \in \mathbb{Z}_p$ in $\mathbb{T}_3$, return~$m'_{(i,j)}$.
\item If $t_i\neq t' \lor j \neq k'$, $m'_{(i,j)} \xleftarrow{\$} \mathbb{Z}_p$, $\mathbb{T}_3 \leftarrow \mathbb{T}_3 \cup \{(t_i, m_j, m'_{(i,j)})\}$, return~$m'_{(i,j)}$.
\item If $t_i = t' \land j = k'$, $\mathbb{T}_3 \leftarrow \mathbb{T}_3 \cup \{(t_i, m_j, m_{\MCL})\}$, return $m_{\MCL}$.
\end{itemize}

\item $\mathcal{O}^{\Sign}(``\Inst", m_j):$
Given an input $(``\Inst", m_j)$, $\B$ performs the following procedure.
\begin{itemize}
\item If $t_{ctr} \notin [T]$, return $\bot$.
\item If $``\Inst" = ``\instskip"$, $t_{ctr} \leftarrow t_{ctr} +1$.
\item If $``\Inst" = ``\instsign"$,
\begin{itemize}
\item If $t_{ctr} \neq t'$, $E \leftarrow X_{\MCL}^{r_{(1,ctr)}} X_{\MCL}^{r_{(2,ctr)}m'_{(ctr,j)}} $ where $r_{(1,i)}$, $r_{(2,i)}$, and $m'_{(i,j)}$ are retreived from $(t_{ctr}, r_{(1,ctr)}, F_{ctr}) \in \mathbb{T}_1$, $(t_{ctr}, r_{(2,ctr)} ,B_{ctr}) \in \mathbb{T}_2$, and $(t_{ctr}, m_j, \allowbreak m'_{(ctr,j)}) \in \mathbb{T}_3$ respectively.
$Q \leftarrow Q \cup \{m_j\}$, return $\sigma_{ctr,j} \leftarrow (E, t_{ctr})$, then update $t_{ctr} \leftarrow t_{ctr} +1$.
\item If $t_{ctr} = t' \land j = k'$, $Q \leftarrow Q \cup \{m_j\}$, return $\sigma_{ctr,j} \leftarrow (E_{\MCL}, t_i)$, then update $t_{ctr} \leftarrow t_{ctr} +1$
\item If $t_{ctr} = t' \land j \neq k'$, abort the simulation.
\end{itemize}
\end{itemize}

\item {\bf Output procedure:} 
$\B$ receives a forgery $((\vk^*_1,\dots, \vk_{r^*}^*), (m_1^*, \dots,\allowbreak m_{r^*}^*), \Sigma^*)$ outputted by $\A$.
Then $\B$ proceeds as follows.

\begin{enumerate}
\item If $\MCLSyncASAggVerify(pp^*, (\vk^*_1,\dots, \vk_{r^*}^*), (m_1^*, \dots, m^*_{r^*}), \allowbreak \Sigma^*) \neq 1$, then abort.
\item If there exists $j \in [r^*]$ such that $\vk^*_j \neq \vk^* \land \vk^*_j \notin L$, then abort.
\item If there is no $j^* \in [r^*]$ such that $\vk^*_{j^*} = \vk^* \land m^*_{j^*} \notin Q$, then abort.
\item Set $j^* \in [r^*]$ such that $\vk^*_{j^*} = \vk^* \land m^*_{j^*} \notin Q$.
\item Parse $\Sigma^*$ as $(E^*{}', t^*)$.
\item \label{tnoabort} If $t^* \neq t'$, then abort.
\item $m^*_{j^*}{}' \leftarrow H_3(t^*,m^*_{j^*})$
\item \label{mnoabort} If $m^*_{j^*}{}' = m_{\MCL}$, then abort.
\item For $i \in [r^*]\backslash \{j^*\}$, retrieve $\sk^*_i = x_i$ of $\vk^*_i$ from $K$.
\item $F' \leftarrow H_1(t^*)$, $B' \leftarrow H_2(t^*)$,
$m'_{i} \leftarrow H_3(t^*,m^*_{i})$ for $i \in [r^*]\backslash \{j^*\}$, \\
$E' \leftarrow E^*{}' \cdot \left(F'{}^{\sum_{i \in [r^*]\backslash \{j^*\}} x_i} B'{}^{\sum_{i \in [r^*]\backslash \{j^*\}} x_i m'_{i}} \right)^{-1}$.
\item Return $(m^*_{\MCL}, \sigma^*_{\MCL}) \leftarrow (m^*_{j^*}, (w_{\MCL},A_{\MCL}, B', C_{\MCL}, D_{\MCL}, E'))$.
\end{enumerate}
\end{itemize}

We confirm that if $\B$ does not abort, $\B$ can simulate the $\rmEUFCMA$ game of the $\MCLSyncAS$ scheme.

\begin{itemize}
\item {\bf Initial setup:}
First, we discuss the distribtuon of $pp^*$.
In the original $\rmEUFCMA$ game of the $\MCLSyncAS$ scheme, $pp^* = (\mathcal{G}, g)$ is constructed by $\mathcal{G}=(p, \G, \G_{T}, e) \leftarrow \mathsf{G}(1^\lambda)$ and $g \xleftarrow{\$} \G^*$.
In the simulation of $\B$, $pp^*$ is a tuple $(\mathcal{G}_{\MCL}, g_{\MCL})$.
This tuple is constructed by $\C$ as $\mathcal{G}_{\MCL}=(p, \G, \G_{T}, e) \leftarrow \mathsf{G}(1^\lambda)$ and $g_{\MCL} \xleftarrow{\$} \G^*$.
Therefore, $\B$ simulates $pp^*$ perfectly.
Next, we discuss the distribution of $\vk^*$.
In the original $\rmEUFCMA$ game of the $\MCLSyncAS$ scheme, $\vk$ is computed by $x \xleftarrow{\$} \mathbb{Z}^*_p$ and $\vk^* \leftarrow g^x$.
In the simulation of $\B$, $\vk^*$ is set by $X_{\MCL}$.
Since $X_{\MCL}$ is computed by $\C$ as $x_{\MCL} \xleftarrow{\$} {Z}_p$ and $X_{\MCL} \leftarrow g^{x_{\MCL}}$, distributions of $\vk$ between the original game and simulation of $\B$ are identical.
Hence, the distributions of $(pp^*, \vk^*)$ are identical.

\item {\bf Output of $\mathcal{O^{\Cert}}$:} 
This is clearly that $\B$ can simulate the original $\rmEUFCMA$ game of the $\MCLSyncAS$ scheme perfectly.
\item {\bf Output of $\mathcal{O}^{H_1}$:} 
In the original game, hash values of $H_1$ is chosen from $\G$ uniformly at random.
In the simulation of $\B$, if $t_i \neq t'$, the hash value $H(t_i)$ is set by $g^{r_{(1, i)}}$ where $r_{(1, i)} \xleftarrow{\$} \mathbb{Z}_p$.
Obviously, in this case, $\B$ can simulate $\mathcal{O}^{H_1}$ perfectly.
If $t_i = t'$, the hash value $H(t_i)$ is set by $F=A_{\MCL}D_{\MCL}^{w_{\MCL}}= A_{\MCL}^{1+y_{\MCL}z_{\MCL}w_{\MCL}}$ where $Y_{\MCL}= g_{\MCL}^{y_{\MCL}}$, $Z_{\MCL}=g_{\MCL}^{z_{\MCL}}$, and $w_{\MCL}$ is chosen by $\C$ as $w_{\MCL} \leftarrow \mathbb{Z}_p$.
For fixed $y_{\MCL} \in \mathbb{Z}^*_p$ and $z_{\MCL} \in \mathbb{Z}_p^*$, the distribution $\alpha$ where $\alpha \xleftarrow{\$} \mathbb{Z}_p$ and $w_{\MCL} \xleftarrow{\$} \mathbb{Z}_p$, $\alpha \leftarrow 1+y_{\MCL}z_{\MCL}w_{\MCL}$ are identical.
This fact implies that $\B$ also simulate $\mathcal{O}^{H_1}$ perfectly in the case of $t_i = t'$.
Therefore, $\B$ simulates $\mathcal{O}^{H_1}$ perfectly.

\item {\bf Output of $\mathcal{O}^{H_2}$:} 
As the same argument of $\mathcal{O}^{H_1}$, if $t_i \neq t'$, $\B$ can simulate hash values $H(t_i)$ perfectly.
In the case of $t_i = t'$, the hash value $H(t_i)$ is set by $B_{\MCL} = A^{y_{\MCL}} = g^{x_{\MCL}y_{\MCL}}$.
For fixed $x_{\MCL} \in \mathbb{Z}^*_p$, the distributions of $B$ where $y_{\MCL} \xleftarrow{\$} \mathbb{Z}^*_p$, $B \leftarrow g^{x_{\MCL}y_{\MCL}}$ and $B \xleftarrow{\$} \G^*$ are identical.
Therefore, $\B$ simulates $\mathcal{O}^{H_2}$ perfectly.


\item {\bf Output of $\mathcal{O}^{H_3}$:} 
If $t_i \neq t' \lor j \neq k'$, clearly $\B$ can simulate $\mathcal{O}^{H_3}$ perfectly.
If $t_i = t' \land j = k'$, the hash value $H_3(t_i, m_j)$ is set by $m_{\MCL}$.
Since $m_{\MCL}$ is chosen by $\B$ as $m_{\MCL} \xleftarrow{\$} \mathbb{Z}_p$,
$\B$ simulates $\mathcal{O}^{H_3}$ perfectly.

\item {\bf Output of $\mathcal{O^{\Sign}}$:} For the sake of argument, we denote $X_{\MCL}= g_{\MCL}^{x_{\MCL}}$ $(x_{\MCL} \in \mathbb{Z}_p^*)$.
If $t_i\neq t'$, $\B$ sets $E \leftarrow X_{\MCL}^{r_{(1,i)}} X_{\MCL}^{r_{(2,i)}m'_{(i,j)}}$ and output the signature $\sigma = (E, t_i)$.
Now we confirm that $\sigma$ is a valid signature on the message $m_j$.
The following equation
\begin{equation*}
\begin{split}
E = X_{\MCL}^{r_{(1,i)}} X_{\MCL}^{r_{(2,i)}m'_{(i,j)}} &= (g^{x_{\MCL}}_{\MCL})^{r_{(1,i)}} (g^{x_{\MCL}}_{\MCL})^{r_{(2,i)}m'_{(i,j)}}\\
&= H_1(t_i)^{x_{\MCL}} H_2(t_i)^{x_{\MCL}m'_{(i,j)}}
\end{split}
\end{equation*}
holds where $m'_{(i,j)} = H_3(t_i, m_j)$.
This fact implies that
\begin{equation*}
e(E,g) = e(H_1(t_i)H_2(t_i)^{m'_{(i,j)}}, \vk^*)
\end{equation*}
holds.
Therefore, $\sigma$ is valid signature on the message $m_j$.

If $t_i\neq t' \land j = k'$, $\B$ sets $E \leftarrow E_{\MCL}$, return $\sigma_{i,j} \leftarrow (E, t_i)$ to $\A$.
We also confirm that $\sigma$ is a valid signature on the message $m_j$.
In the case, $H_1(t_i) = A_{\MCL}D_{\MCL}^{w_{\MCL}}$, $H_2(t_i)= B_{\MCL}$, and $H_3(t_i, m_j) = m'_{(i,j)} = m_{\MCL}$ hold.
Since $E_{\MCL}$ is the valid signature of the $\MCL$ scheme on message $m_{\MCL}$,
\begin{equation*}
\begin{split}
e (E_{\MCL},g) &= e(A_{\MCL}B_{\MCL}^{m_\MCL}D_{\MCL}^{w_{\MCL}}, X_{\MCL})\\
 &= e((A_{\MCL}D_{\MCL}^{w_{\MCL}})B_{\MCL}^{m_\MCL}, X_{\MCL})
\end{split}
\end{equation*}
holds. This implies that
$e(E,g) = e(H_1(t_i)H_2(t_i)^{m'_{(i,j)}}, \vk^*)$ 
where $m'_{(i,j)}= H_3(t_i, m_j)$.
\end{itemize}
By the above discussion, we can see that $\B$ does not abort, $\B$ can simulate the $\rmEUFCMA$ game of the $\MCLSyncAS$ scheme.

Second, we confirm that when $\A$ successfully output a valid forgery $((\vk^*_1,\dots, \allowbreak \vk_{r^*}^*), \allowbreak(m_1^*, \dots, m_{r^*}^*), \Sigma^*)$ of the $\MCLSyncAS$ scheme, $\B$ can forge a signature of the $\MCL$ scheme.
Let $((\vk^*_1,\dots, \vk_{r^*}^*), \allowbreak(m_1^*, \dots, m_{r^*}^*), \Sigma^*)$ be a valid forgery output by $\A$.
Then there exists $j^* \in [r^*]$ such that $\vk^*_{j^*} = \vk^*$.
By the verification equation of $\MCLSyncASVerify$, 
\begin{equation*}
e(E^*{}', g) = e\left(H_1(t^*) , \prod^{r^*}_{i=1} \vk^*_i \right) \cdot e\left(H_2(t^*) , \prod^{r^*}_{i=1} (\vk^*_i)^{m_i^*} \right)  
\end{equation*}
holds where $\Sigma^*= (E^*{}', t^*)$ and $H_3(t^*, m_i^*) = m_i^*{}'$ for $i \in [r^*]$.
If $\B$ does not abort in Step \ref{tnoabort} of {\bf Output procedure}, $t^* = t'$ holds.
This means that $H_1(t^*) = A_{\MCL}D_{\MCL}^{w_{\MCL}}$ and $H_2(t^*)= B_{\MCL}$ hold.
These facts imply that
\begin{equation*}
\begin{split}
E^*{}' &= H_1(t^*)^{\sum^{r^*}_{i=1}\sk^*_i}H_2(t^*)^{\sum^{r^*}_{i=1}m_i^*{}'\sk^*_i}\\
&= \left(A_{\MCL}D_{\MCL}^{w_{\MCL}} \right)^{\sum^{r^*}_{i=1}x^*_i}B_{\MCL}^{\sum^{r^*}_{i=1}m_i^*{}'x^*_i}
\end{split}
\end{equation*}
holds where $\sk^*_i = x^*_i$ is a secret key corresponding to $\vk^*_i$.

By setting $F' \leftarrow A_{\MCL}D_{\MCL}^{w_{\MCL}}$ and $B' \leftarrow B_{\MCL}$, 
\begin{equation*}
\begin{split}
E' &= E^*{}' \cdot \left(F'{}^{\sum_{i \in [r^*]\backslash \{j^*\}} x_i} B'{}^{\sum_{i \in [r^*]\backslash \{j^*\}} x_i m'_{i}} \right)^{-1}\\
& = (A_{\MCL}D_{\MCL}^{w_{\MCL}})^{x^*_{j^*}}B_{\MCL}^{m_{j^*}^*{}'x^*_{j^*}}\\
\end{split}
\end{equation*}

Moreover, $e(A_{\MCL}, Y_{\MCL}) = e(B_{\MCL}, g_{\MCL})$, $e(A_{\MCL}, Z_{\MCL}) = e(C_{\MCL}, g_{\MCL})$, and $e(C_{\MCL}, Y_{\MCL})  \allowbreak = e(D_{\MCL}, g_{\MCL})$ holds.
If $\B$ does not abort in Step \ref{mnoabort} of {\bf Output procedure}, $m^*_{j^*}$ is a not queried message for the signing of the $\rmOTEUFCMA$ game of the $\MCL$ scheme.
Therefore, if $\B$ does not abort and outputs $(m^*_{\MCL}, \sigma^*_{\MCL}) \allowbreak \leftarrow (m^*_{j^*}, (w_{\MCL},A_{\MCL}, B', C_{\MCL}, D_{\MCL}, E'))$, $\B$ can forge a signature of the $\MCL$ scheme.

Finally, we analyze the probability that $\B$ succeeds in forging a signature of the $\MCL$ scheme.
First, we consider the probability that $\B$ does not abort at the simulation of signatures.
$\B$ aborts the simulation of $\mathcal{O^{\Sign}}$ if $t_{ctr} = t' \land j \neq k'$.
The probability that $\B$ succeeds in simulating $\mathcal{O^{\Sign}}$ is at least $1/q_{H_3}$.
Next, we consider the probability that $\B$ aborts in Step \ref{tnoabort} of {\bf Output procedure}.
Since $\B$ chooses the target period $t' \leftarrow [T]$, the probability $t^* \neq t'$ is $1/[T]$.
Finally, the probability that $\B$ aborts in Step \ref{mnoabort} of {\bf Output procedure} is $1/p$.
Let $\Adv^{\EUFCMA}_{\MCLSyncAS, \A}$ be the advantage of the $\rmEUFCMA$ game for the $\MCLSyncAS$ scheme of $\A$.
The advantage of the $\rmOTEUFCMA$ game for the $\MCL$ scheme of $\B$ is 
\begin{equation*}
\Adv^{\OTEUFCMA}_{\MCL, \B} \geq \frac{\Adv^{\EUFCMA}_{\MCLSyncAS, \A}}{T \times q_{H_3}} \left(1- \frac{1}{p}\right). 
\end{equation*}
Therefore, we can conclude the proof of Theorem \ref{MCLSyncASEUFCMA}.
\qed
\end{proof}

By combining Theorem \ref{MCLtoqSDHii} and Theorem \ref{MCLSyncASEUFCMA}, we have the following corollary.
\begin{corollary}\label{MCLSynctoqSDHii}
If the $\rmoMSDHii$ assumption holds, then, in the random oracle model, the $\MCLSyncAS$ scheme satisfies the $\rmEUFCMA$ security in the certified-key model.
\end{corollary}

\section*{Acknowledgement}
A part of this work was supported by Input Output Hong Kong, Nomura Research Institute, NTT Secure Platform Laboratories, Mitsubishi Electric, I-System, JST CREST JPMJCR14D6, JST OPERA, and JSPS KAKENHI 16H01705, 17H01695.
We also would like to thank anonymous referees for their constructive comments.

\bibliographystyle{abbrvurl}
\bibliography{WMS}

\setcounter{tocdepth}{2}
\tableofcontents

\end{document}